\def\year{2017}\relax
\documentclass[letterpaper]{article}
\usepackage{aaai17}
\usepackage{times}
\usepackage{helvet}
\usepackage{courier}
\usepackage{graphicx}
\usepackage{url}
\usepackage{amsmath,amssymb,amsthm,bm}
\newtheorem{lemma}{Lemma}

\newtheorem{proposition}{Proposition}
\newcommand*{\Cdot}{\raisebox{-0.45ex}{\scalebox{1.15}{$\cdot$}}}
\begin{document}
%
\title{Correlated Cascades: Compete or Cooperate}
\author{Ali Zarezade$^*$, Ali Khodadadi$^*$, Mehrdad Farajtabar$^\dagger$, Hamid R. Rabiee$^*$, Hongyuan Zha$^\dagger$\\
$^*$Sharif University of Technology, Azadi Ave, Tehran, Iran\\
$^\dagger$Georgia Institute of Tech., North Ave NW, Atlanta, GA 30332, United States\\
\{zarezade,khodadadi\}@ce.sharif.edu, mehrdad@gatech.edu, rabiee@sharif.edu, zha@cc.gatech.edu
}
\maketitle 
\begin{abstract}
In real world social networks, there are multiple cascades which are rarely independent. They usually compete or cooperate with each other. Motivated by the reinforcement theory in sociology we leverage the fact that adoption of a user to any behavior is modeled by the aggregation of behaviors of its neighbors. We use a multidimensional marked Hawkes process to model users product adoption and consequently spread of cascades in social networks. The resulting inference problem is proved to be convex and is solved in parallel by using the barrier method. The advantage of the proposed model is twofold; it models correlated cascades and also learns the latent diffusion network. Experimental results on synthetic and two real datasets gathered from Twitter, URL shortening and music streaming services, illustrate the superior performance of the proposed model over the alternatives.
\end{abstract}

\section{Introduction}
Social networks and virtual communities play a key role in today's life. People share their thoughts, beliefs, opinions, news, and even their locations in social networks and engage in social interactions by commenting, liking, mentioning and following each other.
 This virtual world is an ideal place for studying social behaviors and spread of cultural norms \cite{Vespignani2012},  contagion of diseases \cite{Barabasi2015},  advertising and marketing \cite{valera2015} and estimating the culprit in malicious diffusions \cite{farajtabar2015back}.
Among them, the study of information diffusion or more generally \emph{dynamics on the network} is of crucial importance and can be used in many applications. 
The trace of information diffusion, virus or infection spread, rumor propagation, and product adoption is usually called \emph{cascades}.

In conventional studies of diffusion networks, individual cascades are mostly considered in isolation, \textit{i.e.}, independent of each other~\cite{gomez2015estimating}. However in realistic situations, they are rarely independent and can be \emph{competitive}, when a URL shortening service become popular the others receive less attention; or \emph{cooperative}, when usage of Google Play Music correlates with that of Youtube due to, for example,  simultaneous arrival of new albums (Fig.~\ref{fig:emp_music}). 

Modeling multiple cascades which are correlated to each other is a challenging problem. Considerable work have done to extend basic diffusion models to competitive case \cite{He2012,Pathak2010,Lu2015}. Meyer \textit{et al.} proposed a probabilistic model for diffusion of competitive or cooperative contagions \cite{Myers2012}. They estimate the probability of a user being infected given a sequence of previously observed contagions. But the main drawback of these models is that they are all discrete time, which limits the flexibility of model. Valera in \cite{valera2015} proposed a continuous time method for modeling competing products but incapable of learning the latent diffusion network and  prone to overfitting. 

\begin{figure}[!t]
\centering
\hspace{-4mm}
\includegraphics[width=0.25\textwidth]{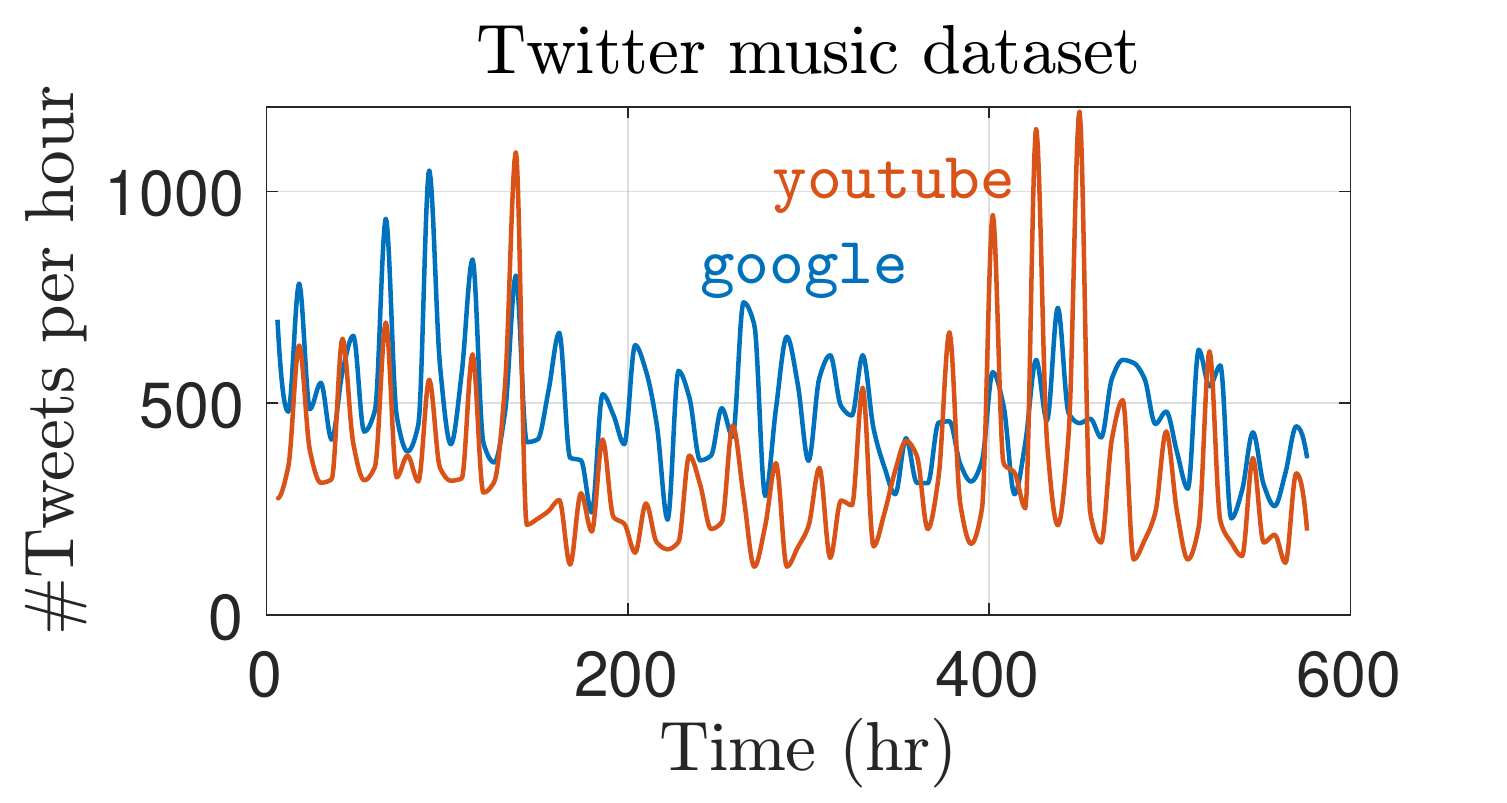}  \hspace{-5mm}
\includegraphics[width=0.25\textwidth]{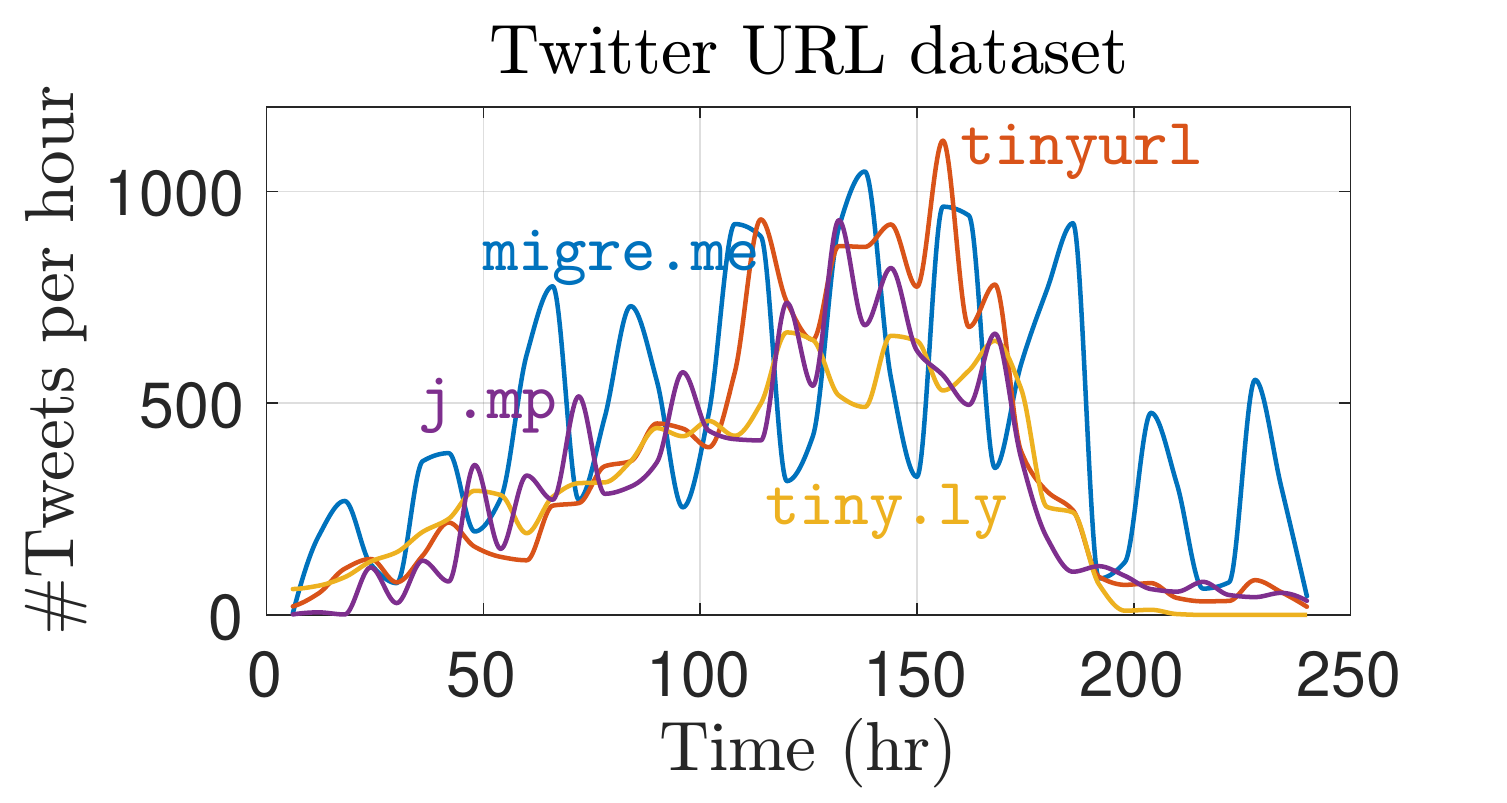} \vspace{-4mm}
\caption{Visualization of correlated cascading behavior in real data.
(\textit{left}) Tweets with terms \texttt{\small{google}} and \texttt{\small{youtube}}  are synchronized most of the time. (\textit{right}) different URL shortening services; \texttt{\small{tiny.ly}} and \texttt{\small{tinyurl}} are cooperating while \texttt{\small{migre.me}} and\texttt{\small{j.mp}} are competing.}
 \vspace{-8mm}
\label{fig:emp_music}
\end{figure} 

Inspired by the sociological evidence in social science about users' behavior, and the success of the recurrent point processes in modeling temporal event histories, we propose a data-driven continuous time method, which can jointly model the spread of multiple correlated behaviors (information, ideas, memes), and learn the latent diffusion network.
Intuitively, the rate (or intensity) that a user adopts a behavior is considered as the weighted sum of those of her neighbors.
The more your friends adopt a behavior the more you are excited to adopt it. 
This users' behavior adoption intensity is modeled by a special stochastic point process, called Hawkes process \cite{Hawkes1971}. 
It is well suited for modeling  temporal events with self-exciting property \cite{Zhou2013a,Blundell2012,Farajtabar2014,xu2016learning}.

Finally, we validate the proposed method on synthetic and two real datasets. First, using synthetic data generated randomly we've studied how effectively we can recover the latent diffusion network parameters. We have also highlighted the correlated behavior versus the independent and the competitive versus cooperative cascades using synthetic data.
Next, we move forward to real data and show the correlated behavior in the real dataset. Furthermore, with parameters learned from training data, we generate future events and compare it with real held-out test data and show that our framework can model the activities in social network better than the alternatives. 
Our contributions are as follows:
\begin{itemize}
\item Modeling the users' behavior adoption by using a multidimensional marked Hawkes process and consequently, the spread of multiple correlated cascades in social networks.
\item Proposing a convex optimization formulation to learn the latent diffusion network and model parameters which is solved in parallel by using the barrier method.
\item Curating a compelling dataset on streaming music services using Twitter API, from tweets of 30,000 active users during one month in 2015 year.
\end{itemize}


\vspace{-3mm}
\section{Prior Works}\label{priorwork}
The spread of information is often modeled as a dynamical process over networks~\cite{Vespignani2012} and its analysis has attracted significant attention in recent years. These studies can be categorized into three groups. 

Early methods studied the information diffusion in continuous time and without any network structure.
They are used mainly to analyze biological contagions \cite{Barabasi2015}. The dynamical process is described by a differential equation which models the number of population in different stages of a disease \cite{Porter2014}.
Heterogeneous mean-field and particle-network frameworks are proposed to remove the homogeneous population assumption and incorporate the network structure \cite{Vespignani2012}. 

Another line of work which is discrete time and considers the network structure, stemmed from sociological theories about influence spread. Typically they assumed that nodes have two states (active, inactive), and are progressive (active nodes can't become inactive).
Linear Threshold and Independent Cascade are two simple and widely studied models of social contagion \cite{Hodas2014}. In Linear Threshold, a node becomes active when the weighted sum of its active neighbors is higher than a prespecified threshold. In Independent Cascade, each infected node has an independent probability to activate its neighbors.

The third category is continuous time and considers the network structure.
Rodriguez et al. proposed a model in \cite{Rodriguez2011,Rodriguez2013} for information diffusion and latent influence network inference using survival theory. They extended it to dynamic networks in \cite{Rodriguez2010structure}. 
The problem of network inference from a set of cascades is theoretically investigated in \cite{Daneshmand2014}. 
In \cite{Iwata2013}, the superposition property of the Poisson process is used to model the effect of users' sharing activities on others in online communities, and consequently learning latent influence network.
Also in \cite{Linderman2014}, the superposition property is used in a fully Bayesian method with parallel inference.
The scalable influence estimation is addressed in \cite{Du2013} by proposing a nearly linear randomized algorithm.
The problem of activity shaping, driving population toward specific target state is investigated in \cite{Farajtabar2014,farajtabar2016multistage}. However, in all of the above models, cascades of adoption/propagation are independent which is usually not true in the real world.

Closely related to the our work, authors in~\cite{farajtabar2015coevolve} proposed a probabilistic framework to model the evolution of information diffusion and network evolution. However, in their work cascades are still evolving independently.
Only recently \cite{valera2015} proposed an algorithm for multiple correlated cascades which models the intensity of user-products by a Hawkes process. 
To model both competition and cooperation it allows the parameters of the intensity function to be negative, and it may results in negative intensity in some cases. Also, it can't learn the latent diffusion network. But we propose a nonlinear user-product intensity using a marked Hawkes process, which has better performance.

\vspace{-2mm}
\section{Proposed Method}\label{proposedmethod}
\subsection{Hawkes Process Background}
A point process is a stochastic process with realizations that are discrete points in time, $\{t_1, t_2, \ldots, t_n\}$.
According to the Kolmogorov extension theorem \cite{Daley2002}, a stochastic process, can be defined using its finite-dimensional distributions. 
To describe the finite-dimensional distributions $f(t_1,t_2,\ldots,t_n)$, we use the chain rule of probability: $f(t_1,t_2,\cdots,t_n)=\prod_i f(t_i|t_{1:i-1})$. Therefore, it suffices to describe only the conditionals, which are abbreviated to $f(t_n|\mathcal{H}_n)$ or simply $f^*(t)$.  Here, $\mathcal{H}_n$ is the history of events before the $n^{th}$ one.
A closely related notion is conditional intensity or rate $\lambda^*(t)$ defined as: 
\begin{equation}\label{equ:hazard}
\lambda^*(t) = {f^*(t)}/[{1-F^*(t)}],
\end{equation}
where $F(\cdot)$ is the cdf of $f(\cdot)$. The relation of $\lambda^*(t)$ and $f^*(t)$ can be expressed in the other direction as in \cite{Aalen2008}:
\begin{align*}
f^*(t) &= \lambda^*(t) \exp\left(-\int_{t_n}^{t} \lambda^*(s) ds\right).
\end{align*}
Another basic concept is the survival function, $S^*(t)=1-F^*(t)$, the probability that no event occurs after the last event in $t_n$ till $t$.
To understand the intensity more intuitively we incorporate the alternative way of describing a point process, the counting process $N$ associated to $\lambda^*(t)$. Let $N(t,s]$ denotes the number of events in interval $(t,s]$. Multiplying both sides of \eqref{equ:hazard} by $dt$ results in: 
\begin{align*}
\lambda^*(t) dt &= \frac{\text{Pr}\left\{N(t_n,t]=0, N(t,t+dt]=1 | \mathcal{H}_n \right\}} {\text{Pr}\left\{N(t_n,t)=0 |\mathcal{H}_n \right\}} \nonumber \\
&=\text{Pr}\left\{N(t,t+dt]=1 | \mathcal{H}_n, N(t_n,t]=0\right\} \nonumber \\
&=\text{Pr}\left\{N(dt)=1 | \mathcal{H}_{t^-}\right\} 
 \approx\mathbb{E}\left[ N(dt) | \mathcal{H}_{t^-} \right]
\end{align*}
where $N(dt):=N(t,t+dt]$ and $\mathcal{H}_{t^-}$ is the history of all events up to $t$.
Different point processes can be constructed by specifying $f^*(t)$ or equivalently $\lambda^*(t)$. 
In the Hawkes process the intensity is dependent on the history:
\begin{align*}
\lambda^*(t) = \mu + \int_{-\infty}^t g(t-s) N(ds) = \mu + \sum_{i=1}^{|\mathcal{H}_{t^-}|} g(t-t_i)
\end{align*} 
where $\mu$ is the base intensity and $g(t)$ is the kernel which is usually exponentially decaying to diminish the effect of past events. Generally, in multidimensional Hawkes process:  
\begin{align*}
\bm{\lambda}^*(t) &= \bm{\mu} + \int_{-\infty}^t \mathbf{A} g(t-s) \bm{N}(ds),
\end{align*}
where $\bm{\lambda}, \bm{\mu}, \bm{N}$ are vectors and $\mathbf{A} = [\alpha_{ij}]$ is a matrix of mutual-excitation kernels. 
$\alpha_{ij}$ parameterizes the influence of user $j$ to user $i$.
The intensity function can be also generalized to the marked case \cite{Hawkes1971}, which a mark $p$, often a subset of $\mathbb{N}$ or $\mathbb{R}$, is associated with each event.
\begin{align*}
\lambda^*(t,p) = \lambda^*(t) f^*(p|t),
\end{align*}
where $f^*(p|t)$ is the conditional mark density function. In the sequel, we omit the star superscript of intensity for notational simplicity.
The mutually-exciting property of the Haweks process makes it a common modeling tool in applications like seismology, epidemiology, reliability, and social network analysis \cite{farajtabar2015back}.

\vspace{-2mm}
\subsection{Correlated Cascades Model}
Suppose we are given a directed network $\mathcal{G}=(\mathcal{V},\mathcal{E})$, with $\left\vert\mathcal{V}\right\vert=N$ nodes and $M$ behaviors (cascades). Nodes of the network can adopt at most one of them in any time.
We denote the \emph{user behavior adoption} by  $\mathcal{D} = \{(t_i, u_i, p_i)\}_{i=1}^{K}$, where each triple $(t_i, u_i, p_i)$ means that user $u_i$ has adopted behavior $p_i$ at time $t_i$. 
We can also define the observations related to user $v$ and behavior (product) $q$ up to time $s$, as $\mathcal{D}_v^q(s) = \left \{ (t_i,u_i,p_i) \in \mathcal{D} \vert t_i<s, u_i=v, p_i=q \right\}$, and define $\mathcal{D}_v(s)$, $\mathcal{D}^q(s)$ and $\mathcal{D}(s)$ in a similar way.

Now the question is, how users in a network decide to adopt a behavior? This is an important question in sociology which has been investigated for decades \cite{Granovetter1973}.
According to social reinforcement theory, the behavior of a user is influenced by her friends \cite{Mcadam1993}. Moreover, each user has behavioral biases  \cite{Farajtabar2014,farajtabar2016multistage}.
These two mechanisms can be well modeled by the Hawkes process.
The mutually-exciting property of the Hawkes process can model the social reinforcement and the base rate can model the bias. Also, the time decaying kernel reflects the diminishing effect of past events.
So, we model the behavior adoption intensity of the user $u$ by: 
\begin{align}\label{equ:intensity}
\lambda_u(t) = \underbrace{\mu_u}_{\text{bias}} + \underbrace{\sum_{i=1}^{\vert\mathcal{D}(t)\vert}{  \alpha_{u_i u} \;e^{-(t-t_i)} } }_{\text{social reinforcement}}
\end{align}
where $\mu_u$ is the base intensity of user $u$ or bias, $\alpha_{ji}$, an element of the latent diffusion network, is the influence of user $j$ on $i$,
and the summation is over the elements of the set $\mathcal{D}(t)$.
The type of adopted behavior can be seen as the mark of the Hawkes process. Therefore, the intensity of user $u$ to adopt product or behavior $p$ is modeled by:
$
\lambda_u(t,p) = \lambda_u(t) f_u(p|t),
$
where $f_u(p|t)$ is the probability that user $u$ adopts behavior $p$ at time $t$ given history $\mathcal{D}(t)$.
To model the mark probability we define the tendency of user $u$ to adopt behavior $p$ as:
\begin{align}\label{equ:excitationfunc}
g_u^p(t) = \mu_u^p + \sum_{i=1}^{\vert\mathcal{D}^p(t)\vert}{  \alpha_{u_i u} \;e^{-(t-t_i)} }.
\end{align}
Intuitively when a user decides to select a behavior she picks the one with maximum tendency among the different behaviors, $\arg\max_p g_u^p(t)$. The probabilistic version of the max function is the \textit{soft-max} function, so we propose
\begin{align}\label{equ:mark}
f_u(p|t) = \frac{\exp(\beta g_u^p(t))}{\sum_q \exp(\beta g_u^q(t))},
\end{align}
as  the nonlinear mark function, where, hyperparameter $\beta$ tunes the mark function.
In the fully competitive case where $\beta\rightarrow\infty$, it converges to deterministic max function, and in the fully cooperative where $\beta \rightarrow 0$, it converges to the uniform density function. 
In the case of linear mark function:
\begin{align}\label{equ:mark-linear}
f_u(p|t) = \frac{ g_u^p(t) }{\sum_q  g_u^q(t)}
\end{align}
the user behavior intensity simplifies to
$\lambda_u(t,p) = g_u^p(t)$.
By decomposing the model likelihood to the product of cascade likelihoods, we can show that it reduces to the independent cascade model \cite{Rodriguez2011}. 
To find the observation likelihood of the proposed model we use the following proposition.
\begin{proposition}
\label{prop1}
For $u=1,2,\ldots,N$, let $N_u$ be a multi-dimensional marked point process on $[0,T]$ with associated intensity $\lambda_u(t)$, and mark density $f_u(p|t)$. Let 
$\mathcal{D}=\{(t_i,u_i,p_i)\}_{i=1}^K$ be a time, user and mark realization of the process over $[0,T]$. Then the likelihood of $\mathcal{D}$ the multidimensional Hawkes process model with
 mutually-exciting parameter $\bm{A} = [\alpha_{ij}]$ and baseline parameter $\bm{\mu} = [\mu_i^p]$, $(i,j=1,2,\cdots,N, p=1,2,\cdots,M)$ is:
\begin{align*}
\mathcal{L}(\theta \vert \mathcal{D}) = \left[\prod_{i=1}^K \lambda_{u_i}(t_i) f_{u_i}(p_i|t_i) \right] \exp\left(-\int_0^T \sum_{u=1}^N \lambda_u(s) ds \right)
\end{align*}
where $\theta = (\bm{\mu},  \bm{A})$ represents the model parameters.
\end{proposition}
\noindent \textit{Proof}.
Using chain rule, the probability of observation is:
\begin{align*}
\mathcal{L}(\theta \vert \mathcal{D}) := f(\mathcal{D} \vert \theta) = \prod_{i=1}^K f\left((t_i,u_i,p_i)| \mathcal{D}(t_i)\right) \prod_{u=1}^{N} S(T,u) 
\end{align*}
where $t_0=0$ and $S(T,u)$ is the probability that the process $\lambda_u(t)$ survive after its last event:
\begin{align}
S(T,u) = \exp\left(-\int_{t_{\vert\mathcal{D}_u\vert}}^T  \lambda_{u}(s) ds\right) \nonumber
\end{align}
by decomposing the probability of observation we have:
\allowdisplaybreaks
\begin{align*}
&f(\mathcal{D} \vert \theta) 
= \prod_{u=1}^{N}	\prod_{i=1}^{\vert\mathcal{D}_u\vert} f\left((t_i,u_i,p_i)| \mathcal{D}(t_i)\right) \prod_{u=1}^{N} S(T,u)  \\
&= \prod_{u=1}^{N}	\prod_{i=1}^{\vert\mathcal{D}_u\vert} \lambda_{u}(t_i) \exp\left(-\int_{t_{i-1}}^{t_i} \lambda_{u}(s) ds\right) f_{u}(p_i | t_i)
 \prod_{u=1}^{N} S(T,u)  \\
&= \prod_{u=1}^{N} \exp\left(-\int_0^{t_{\vert\mathcal{D}_u\vert}} \lambda_{u}(s) ds\right)  \prod_{i=1}^{\vert\mathcal{D}_u\vert} f_{u}(p_i | t_i) \lambda_{u}(t_i)  \prod_{u=1}^{N} S(T,u)  \\
&= \prod_{u=1}^{N} \exp\left(-\int_0^{t_{\vert\mathcal{D}_u\vert}}  \lambda_{u}(s) ds\right) S(T,u)
\prod_{i=1}^{\vert\mathcal{D}_u\vert} f_{u}(p_i | t_i) \lambda_{u}(t_i) \\
&= \prod_{u=1}^{N} \exp\left(-\int_0^T  \lambda_{u}(s) ds\right) 
\prod_{i=1}^{\vert\mathcal{D}_u\vert} f_{u}(p_i | t_i) \lambda_{u}(t_i) 
\\
&= \prod_{u=1}^{N} \exp\left(-\int_0^T  \lambda_{u}(s) ds\right) 
\prod_{u=1}^{N} \prod_{i=1}^{\vert\mathcal{D}_u\vert} f_{u}(p_i | t_i) \lambda_{u}(t_i) \quad\Box 
\end{align*}
According to this proposition and relations (\ref{equ:intensity})-(\ref{equ:mark}), the log-likelihood of the model can be written as:
\begin{align*}
\log\mathcal{L} & (\theta \vert \mathcal{D})  =\sum_{i=1}^{|\mathcal{D}|}\log  \lambda_{u_i}(t_i) - \sum_{u=1}^N\int_0^T \lambda _u(s) ds  \nonumber \\
&+\sum_{i=1}^{|\mathcal{D}|}  \beta g_{u_i}^{p_i}(t_i) 
-\sum_{i=1}^{|\mathcal{D}|} \log \left(  \sum_{q=1}^M \exp\left(\beta g^q_{u_i}(t_i)\right) \right)
\end{align*}
where $\beta$ is the hyper-parameter. We can decompose the summation over $t_i\in\mathcal{D}$ into the summation over $u$ and $t_i\in\mathcal{D}_u$, which shows that the log-likelihood can be decomposed to sum of users log-likelihood:
\begin{align*}
\log\mathcal{L} (\theta \vert \mathcal{D}) = \sum_{u=1}^N \log\mathcal{L} (\theta_u \vert \mathcal{D}_u)
\end{align*} 
where the parameters of user $u$, $\theta_u$ is composed of $\bm{A}_u = [\alpha_{u\Cdot}]$ and $\bm{\mu}_u=[\mu_{u}^{\Cdot}]$. Moreover the user's log-likelihood is: 
\begin{align*}
\log\mathcal{L}  &(\theta_u \vert\mathcal{D}_u) =\sum_{i=1}^{|\mathcal{D}_u|}\log  \lambda_{u}(t_i) - \int_0^T \lambda _u(s) ds  \nonumber \\
&+\sum_{i=1}^{|\mathcal{D}_u|}  \beta g_{u}^{p_i}(t_i) 
-\sum_{i=1}^{|\mathcal{D}_u|} \log \left(  \sum_{q=1}^M \exp\left(\beta g^q_u(t_i)\right) \right).
\end{align*}
\begin{lemma}
		\label{lem1}
		Function $f:\mathbb{R}^n \rightarrow \mathbb{R}$ with $\textnormal {dom}f = \mathbb{R}^n$ 	is convex.
		\begin{align*}
		f(x)=\log \sum_i \exp(a_i^Tx+b_i)
		\end{align*}		
\end{lemma}
\begin{proof}
	Let $A=\left [a_1,a_2,\cdots,a_n  \right ]^T$, $b=\left [b_1,b_2,\cdots,b_n \right ]^T$ and $z_i=\exp(a_i^Tx+b_i)$, using chain rule we have:
	\begin{align*}
	\nabla^2 f(x) &= A^T \left( \frac{1}{\bm{1}^T z} \text{diag}(z)-\frac{1}{(\bm{1}^Tz)^2} z z^T \right ) A
	\end{align*}

	Now we must show that for all $u$ we have $u^T\nabla^2 f(x) u \geq 0$, or equivalently for all $v$, where $v =Au$ we have:
	\begin{align*}
	v^T \left( \frac{1}{\bm{1}^T z} \text{diag}(z)-\frac{1}{(\bm{1}^Tz)^2} z z^T \right ) v & \geq 0 \\
	\frac{\sum_i v_i^2 z_i }{\sum_i z_i} - \left ( \frac{\sum_i v_i z_i }{\sum_i z_i} \right )^2 & \geq 0 
	\end{align*}
	which holds according to Cauchy-Schwarz inequality.
\end{proof}
\begin{figure}
\centering
\includegraphics[width=0.169\textwidth]{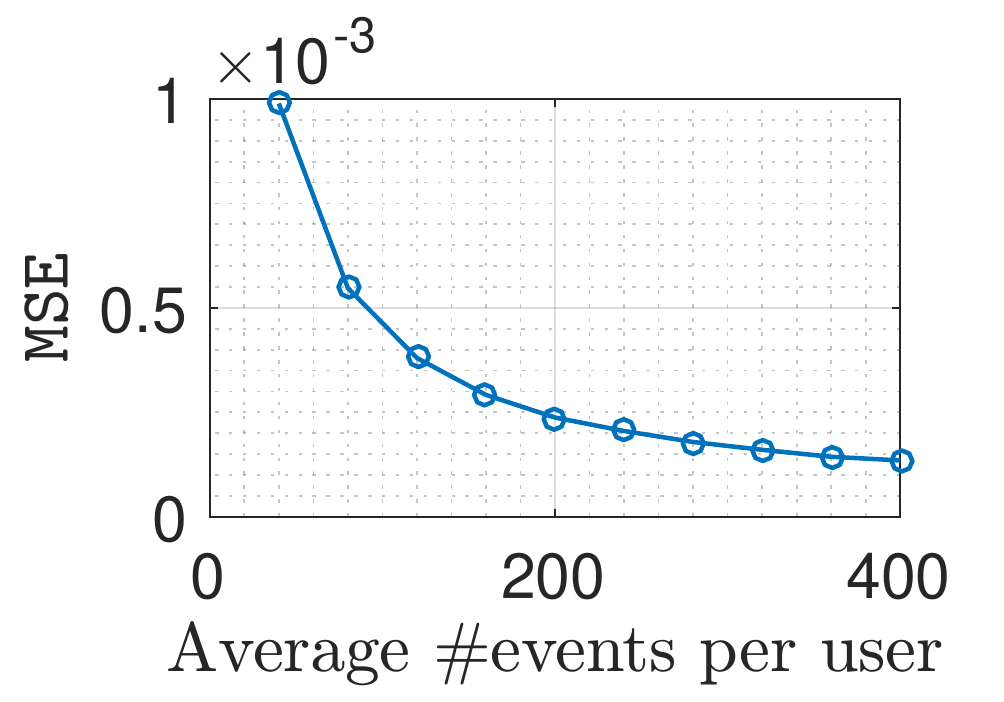}\hspace{-3mm}
\includegraphics[width=0.160\textwidth]{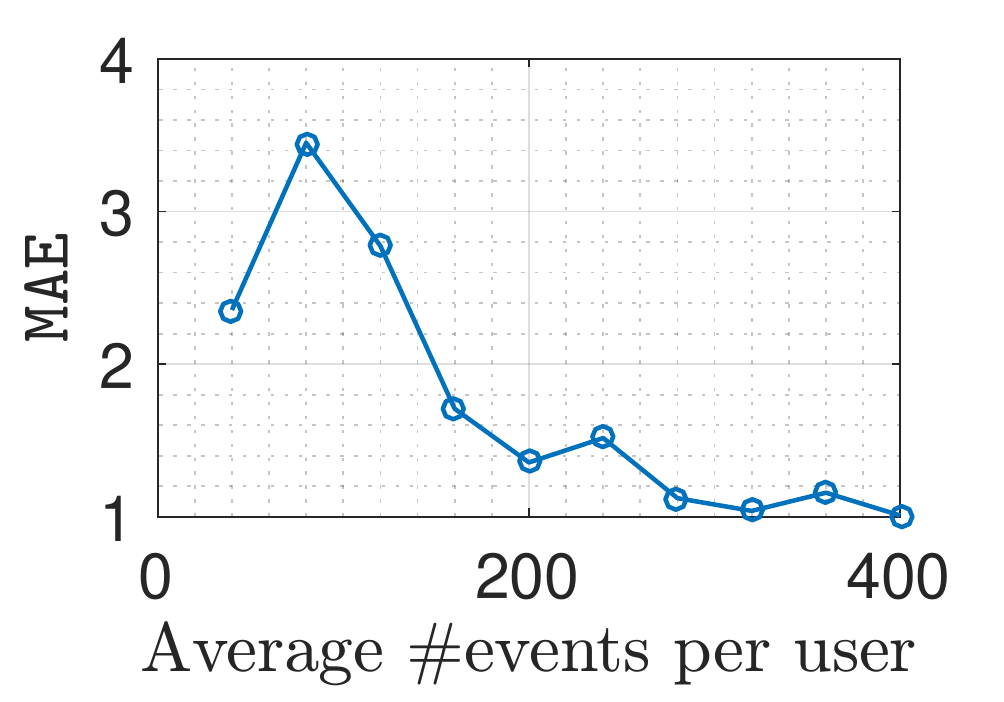}\hspace{-3mm}
\includegraphics[width=0.160\textwidth]{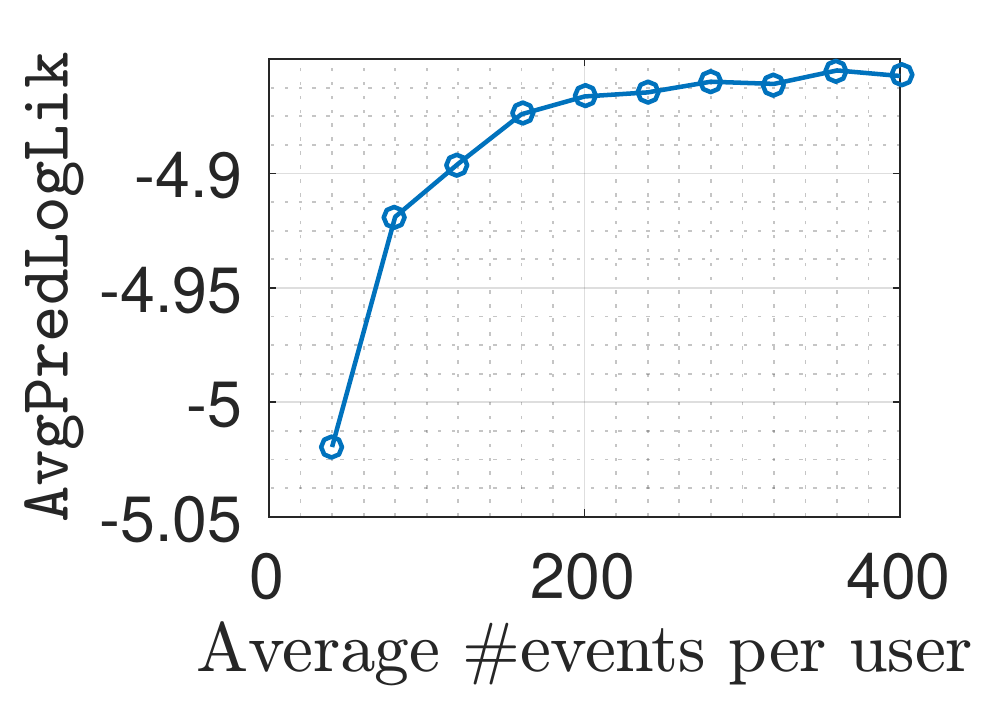}
\vspace{-4mm}
\caption{Performance of the parameter learning on synthetic data.}
\label{fig:synthetic01}
\vspace{-4mm}
\end{figure}

\begin{figure*}
\centering
  \begin{tabular}{ccccc}
   &\small Independent & \small Correlated, $\beta=0.1$ & \small Correlated, $\beta=1$ &  \small Correlated, $\beta=100$ \\ 
  \rotatebox{90}{\hspace{1.2cm}\small Intensity}  
&\includegraphics[width=0.21\textwidth]{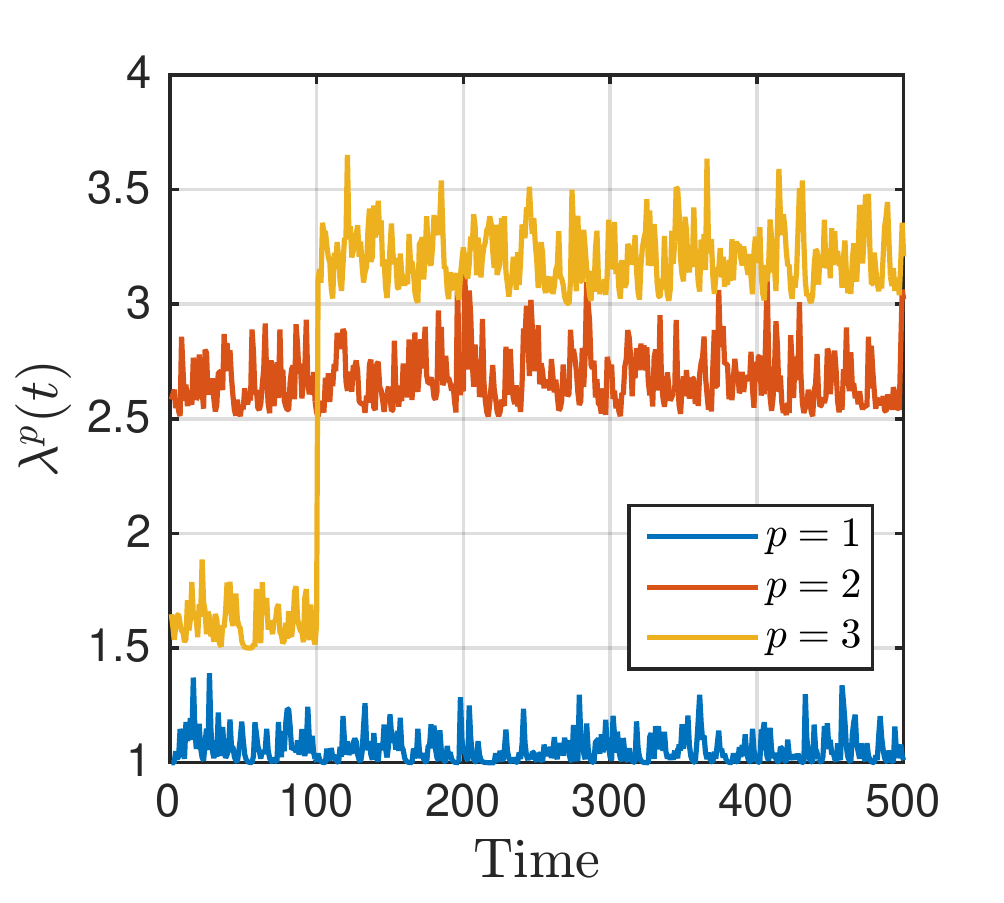}
&\includegraphics[width=0.21\textwidth]{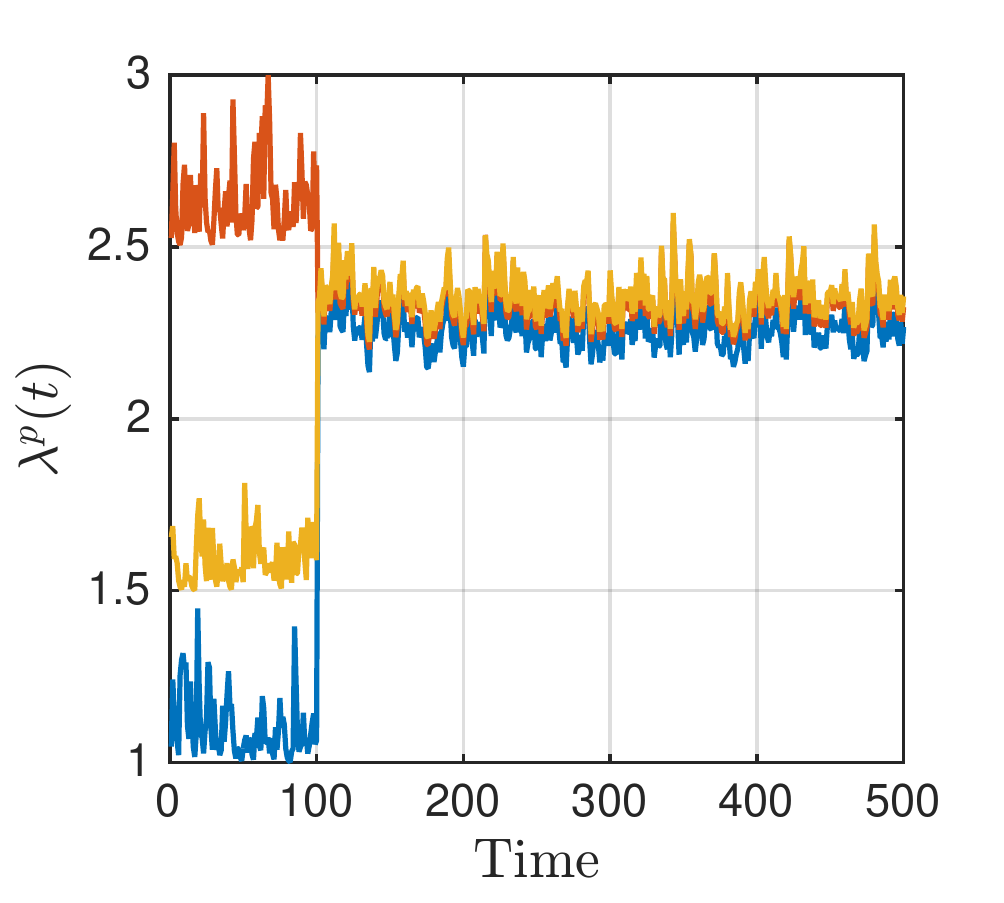}
&\includegraphics[width=0.21\textwidth]{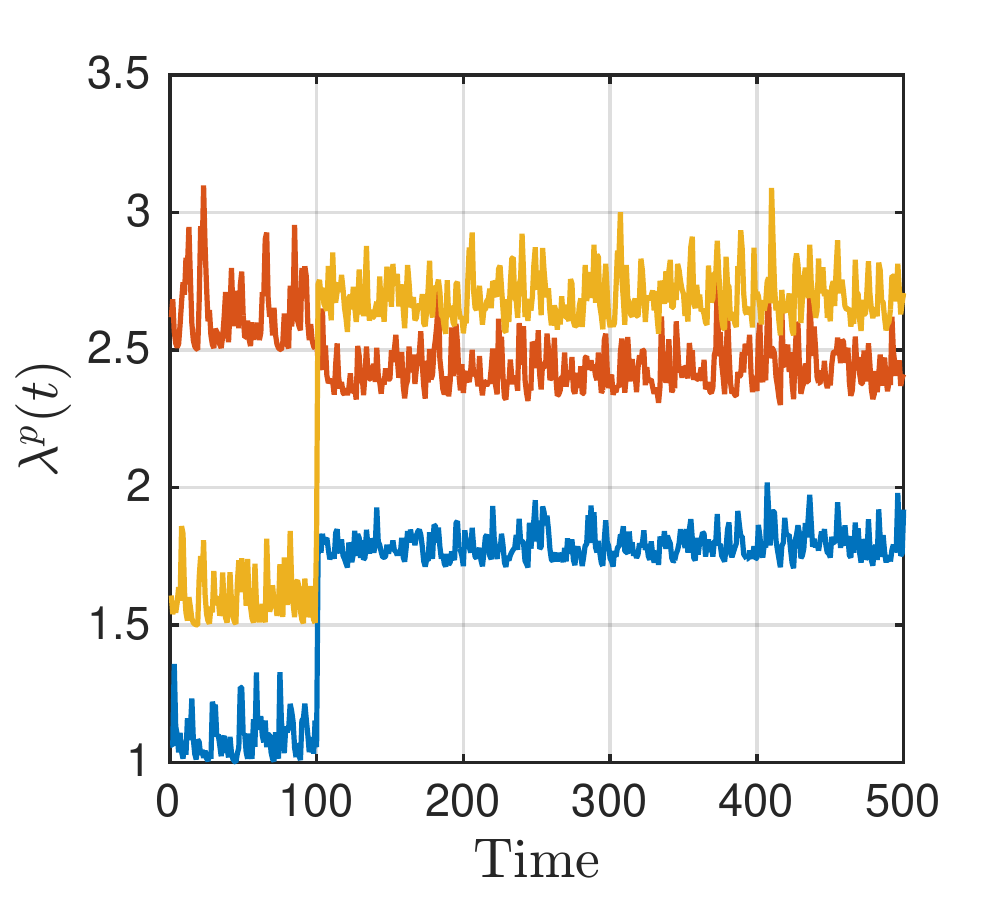}
&\includegraphics[width=0.21\textwidth]{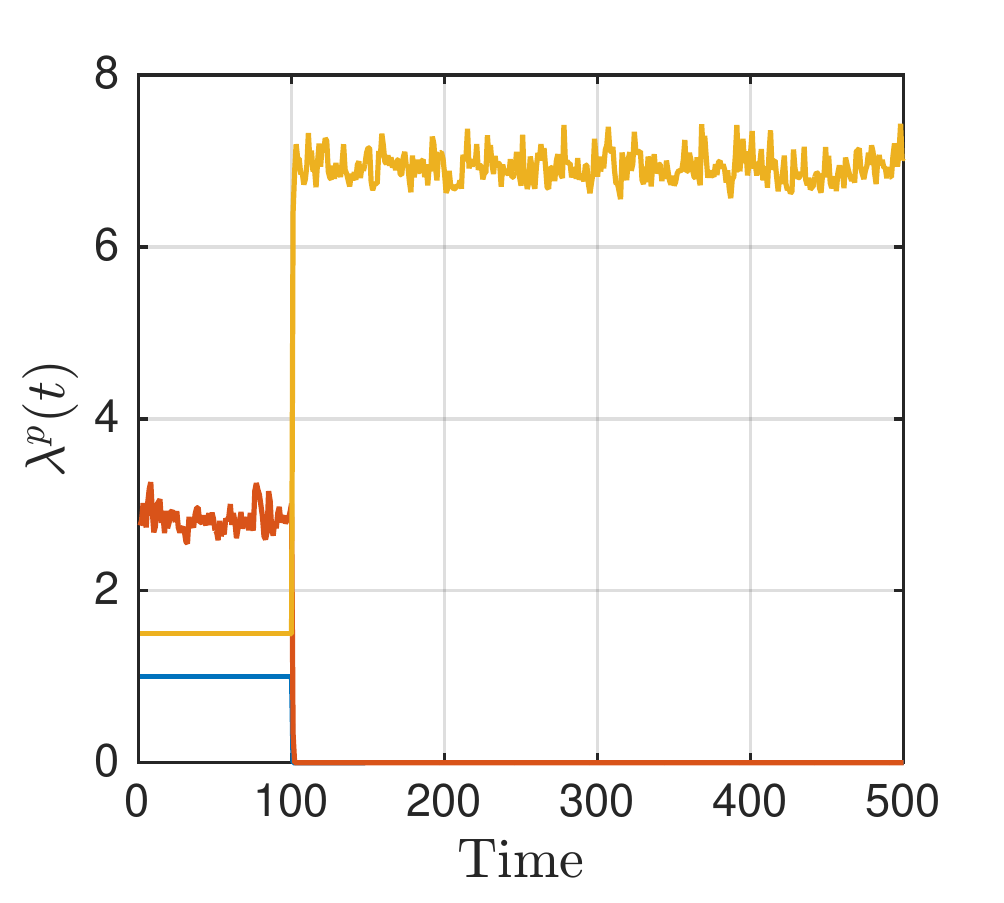}\\
 \rotatebox{90}{\hspace{1.2cm}\small Market Share}
&\includegraphics[width=0.21\textwidth]{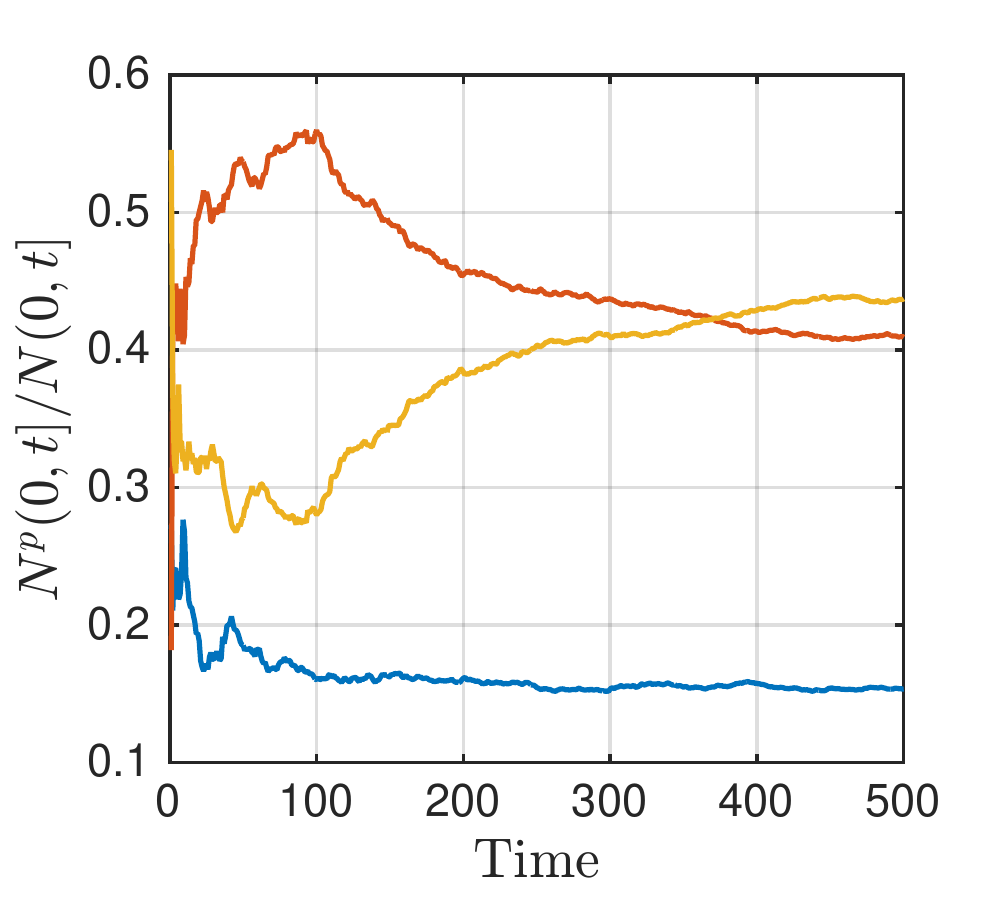}
&\includegraphics[width=0.21\textwidth]{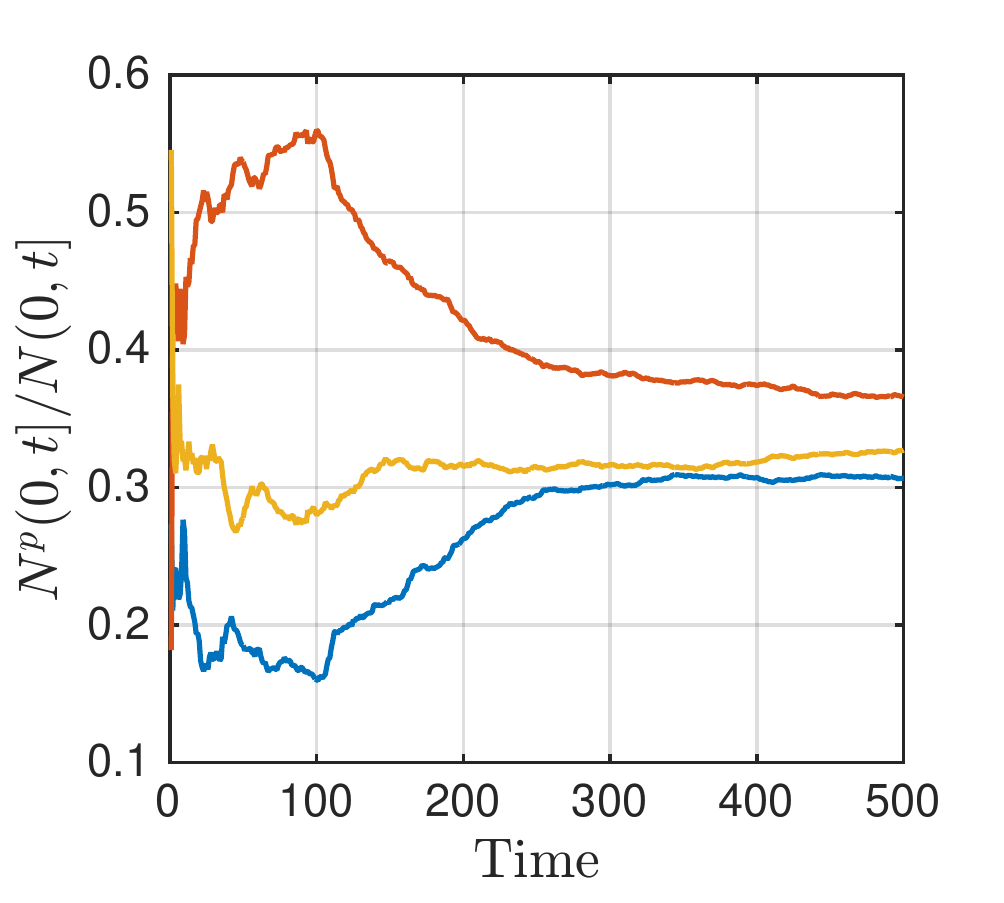}
&\includegraphics[width=0.21\textwidth]{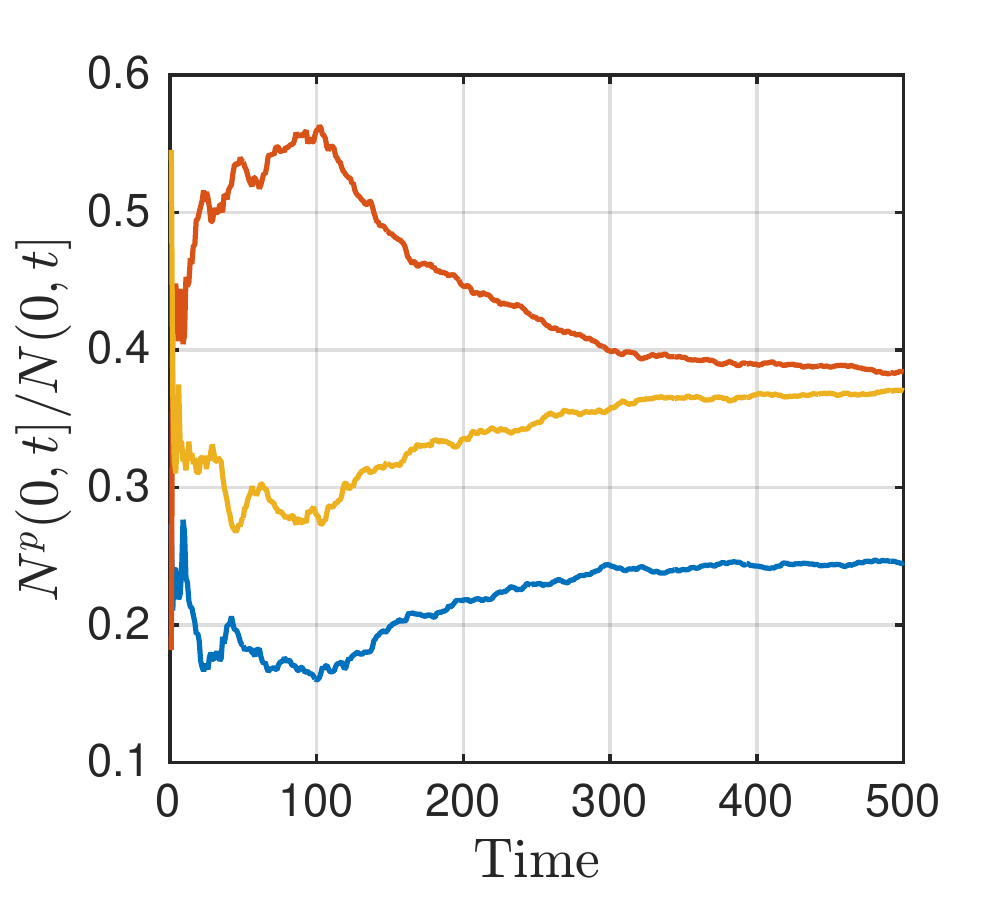}
&\includegraphics[width=0.21\textwidth]{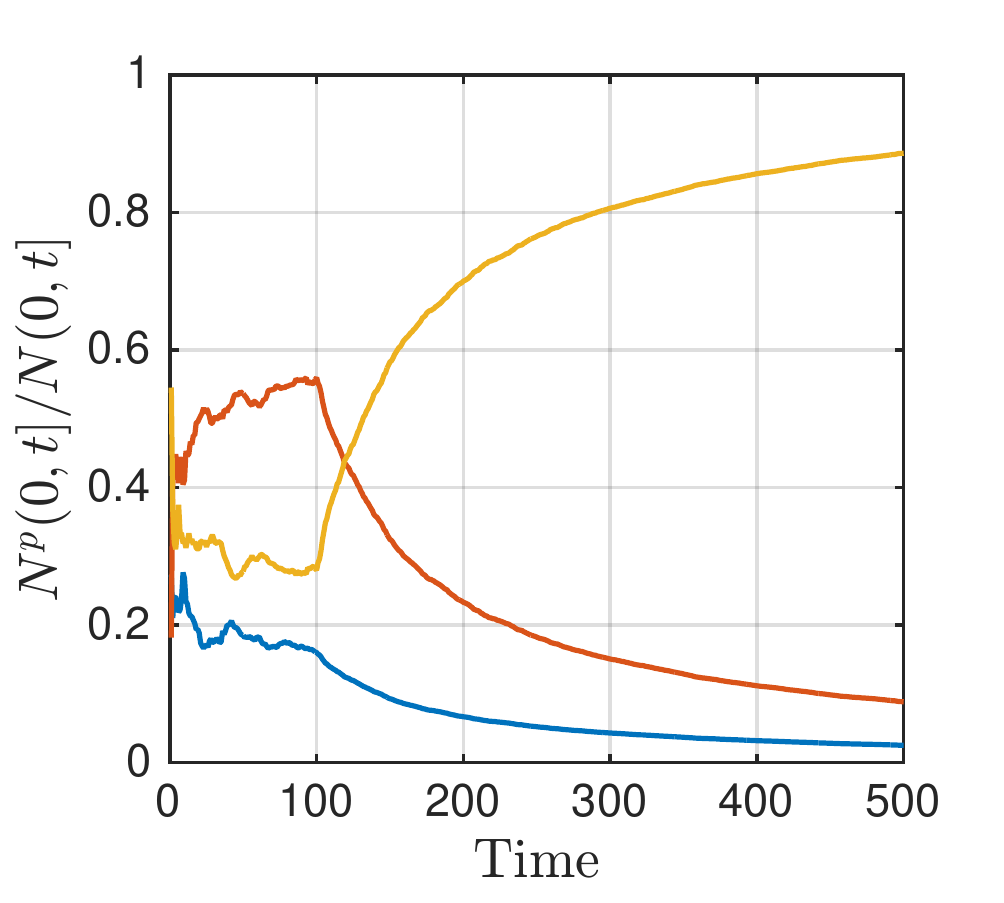} 
 \end{tabular} \vspace{-5mm}
\caption{Intensity and market share for independent and different correlated models. In correlated models, after incentivization in time 100, the product usage of the other products also change, whereas in independent models they remain intact.}
\vspace{-4mm}
\label{fig:synthetic02}
\end{figure*}
We continue with the following proposition which establishes the tractability of parameter learning and allows us to identify the model efficiently.

\begin{proposition}
The negative of the log-likelihood function, $-\log\mathcal{L} (\theta_u \vert \mathcal{D}_u)$ is convex.
\end{proposition}
\begin{proof}
The first term is the negative log of a linear function which is convex, according to composition rules. The second and third term are linear, and the fourth term is convex according to lemma \ref{lem1}.
\end{proof}
Now to find the model parameters we can use the maximum likelihood estimation:
\begin{align}\label{equ:cvx}
& \underset{\theta}{\text{minimize}} \,  - \log\mathcal{L} (\theta_u \vert \mathcal{D}_u) \nonumber
& \text{subject to }  \, \theta \geq 0
\end{align}
on each user, which has unique solution according to proposition 2 and can be solved in parallel for different users. 

\vspace{-2mm}
\section{Experiments\footnote{Implementation codes and datasets can be found at \url{https://github.com/alikhodadadi/C4}.}}\label{experiments}
\vspace{-2mm}
\subsection{Synthetic Data}
We first explain how to generate the synthetic data, then introduce the evaluation criteria. Afterward, we describe the setting for learning the model parameters. Finally, the performance of the algorithm and an experiment that is designed to show the prosperity of the correlated model with respect to its independent version is investigated.

\textbf{Dataset Preparation}.
We generated a random network with $N=50$ and $M=5$. The parameters of the models were drawn randomly from uniform distribution $\mu_{i,p} \sim U(0,0.1)$ and $\alpha_{i,j} \sim U(0,0.01)$. Also, we set $\beta=1$. Then we sampled 20,000 train events and 2000 test events from the proposed model using the thinning method \cite{Ogata1981}. 
The convex optimization is solved in parallel using the Barrier method which  transforms a constrained convex optimization to an unconstrained one.
 
\textbf{Evaluation Criteria}.
We evaluated the accuracy of learning the model parameters using 
\texttt{MSE}, the average squared error between the estimated and true parameters;
\texttt{MAE}, the averaged relative error between the  estimated and true parameters; and \texttt{AvgPredLogLik}, the negative log-likelihood over unseen test events, divided by the number of test events.

\textbf{Parameter Learning}.
We trained 10 models, on $10\%$ to $100\%$ of the synthetic training data. 
In Fig \ref{fig:synthetic01}, we have evaluated the parameter learning and reported three accuracy measures . To be compatible with the real dataset, we have plotted the measures with respect to the average number of events per user.
As expected, with the increase in number of training events the accuracy of recovering the parameters improves. 

\textbf{Correlated Cascades}. 
We also designed an experiment to compare correlated cascade with independent cascade model. We form the independent cascade model using the linear mark (Eq. \ref{equ:mark-linear}), instead of the exponential mark (Eq. \ref{equ:mark}). Then randomly generate 4 similar models with the same $\mu_{i,p}$ and $\alpha_{i,j}$, where $\alpha_{i,j} \sim U(0,0.1)$, and $\mu_{i,p}$ for all users were generated with small noise around $0.2$, $0.5$ and $0.3$ for product $p=1,2$ and $3$, respectively. The number of nodes and products are also set to $N=50$, and $M=3$, respectively. In the correlated models, we set $\beta=0.1,1,100$ to see the effect of mark function on the competitive or cooperative behavior of the proposed model. 
To show the success of our method in generating the correlated cascades, we design a simple incentivization scenario. For all models, the history of  events before time $100$ is generated by the independent model. Then the parameters $\mu_u^p$ of product $p=3$ for all users is doubled, which can be regarded as an incentivization of users by the third service provider. Afterward, each model generates its events separately.
In Fig. \ref{fig:synthetic02} the overall intensity of all users for each product, $\lambda^p(t)$ and the cumulative market share of each product, ${N^p(0,t]}/{\sum_q N^q(0,t]}$ is illustrated.
From the intensity diagram we can see that, after the incentivization, the intensity of users in the correlated model with $\beta=0.1$, becomes approximately the same. But in highly competitive model with $\beta=100$, the third product is dominated shortly after the incentivization which can be seen also from the market share diagrams of Fig \ref{fig:synthetic02}.
To better understand differences between the independent and correlated model, note the intensity of independent model, row and column one of Fig \ref{fig:synthetic02}. In independent model, by incentivizing product $p=3$, its intensity increases but the other two product  are not influenced by this change. 
Also the adoption of incentivized product is increased, but the intensity diagram clearly shows that the other two product are not affected by this change. On the contrary in the correlated models, this change affects on the usage of all other products, which validates the correlated nature of our model. 

\vspace{-1mm}
\subsection{Real Data}
\vspace{-1mm}
In this section, we introduce the real datasets, then explain the evaluation criteria and the settings for parameter learning. Finally, we present the results and the comparisons.

\textbf{Datasets Preparation}.
We use the data crawled from Twitter \cite{Hodas2014}.
This dataset is composed of 213K tweets which contain URLs that shortened by URL shoehorning services. The data was collected over three weeks in Fall of 2010 and is comprised of almost 2K distinct URLs. 
We post-process this dataset by first finding the six most popular ULR shortening services, which are {\small\texttt{bit.ly}, \texttt{migre.me}, \texttt{tinyurl.com}, \texttt{tiny.ly}, \texttt{j.mp}}, and {\small\texttt{is.gd}}. Then, we select a collection of tweets of about 1000 users with at least 100 tweets which contain any of the mentioned URLs. We refer to this dataset by ``Twitter URL dataset".
\begin{figure}
\centering
\includegraphics[width=0.23\textwidth]{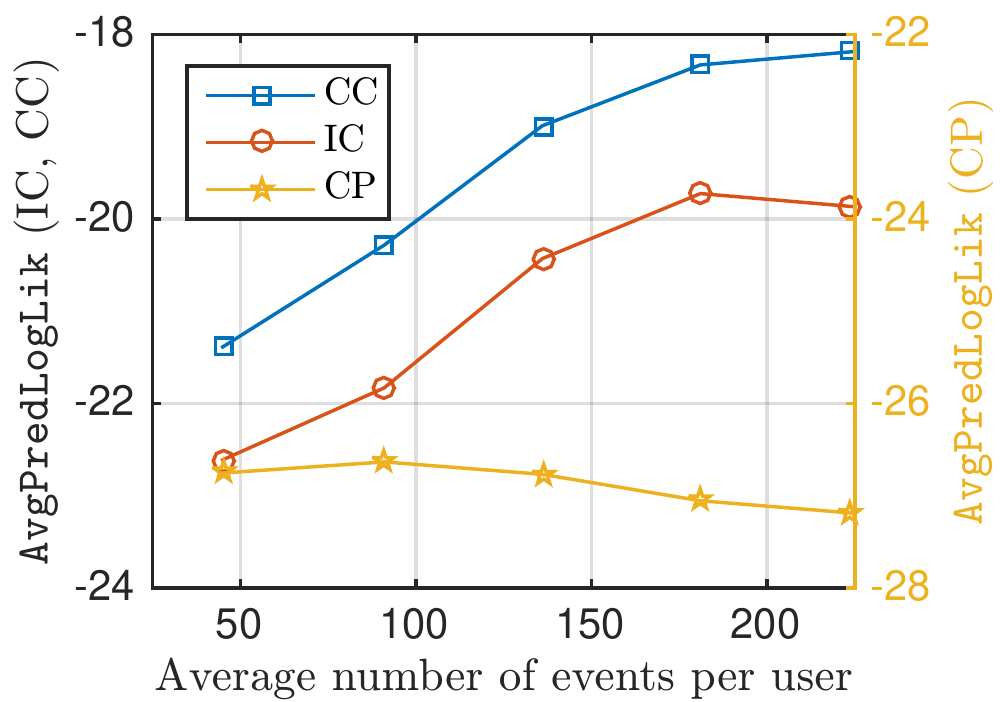}
\includegraphics[width=0.23\textwidth]{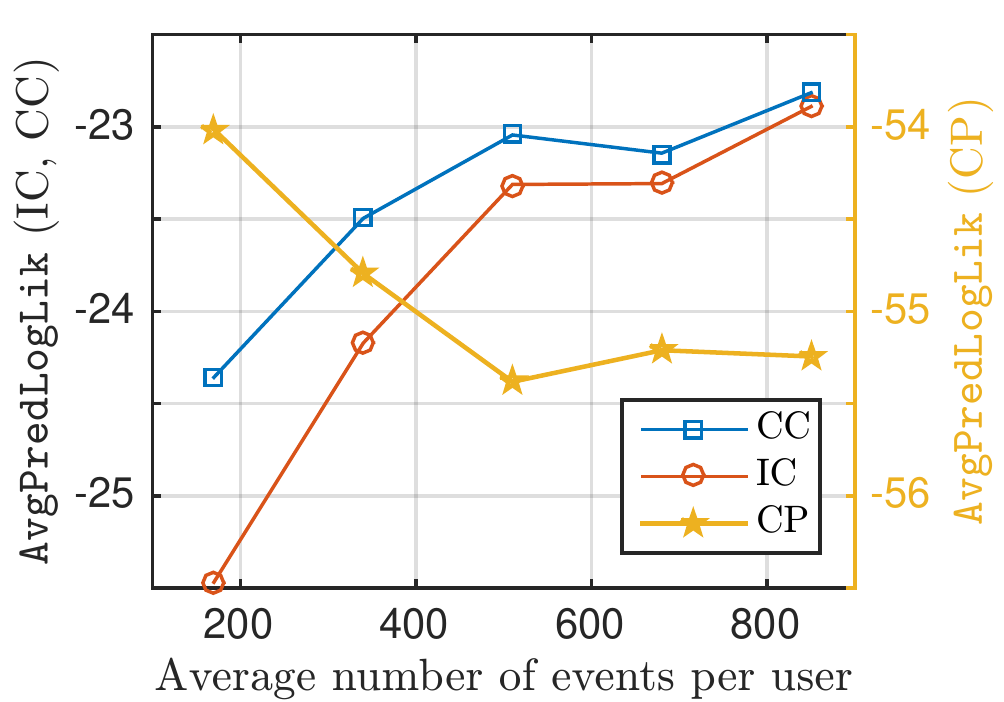} \vspace{-2mm}
\caption{Performance of parameter learning reported via average negative log likelihoods on real datasets, for different size of training set, in Twitter (\textit{left}) URL and (\textit{right}) music datasets.
\texttt{CP} is overfitted and the generalization power of the proposed method is more than \texttt{IC}.
\vspace{-3mm}
 }
\label{fig:reallglks}
\end{figure}
\begin{figure}
\centering
\includegraphics[width=0.245\textwidth]{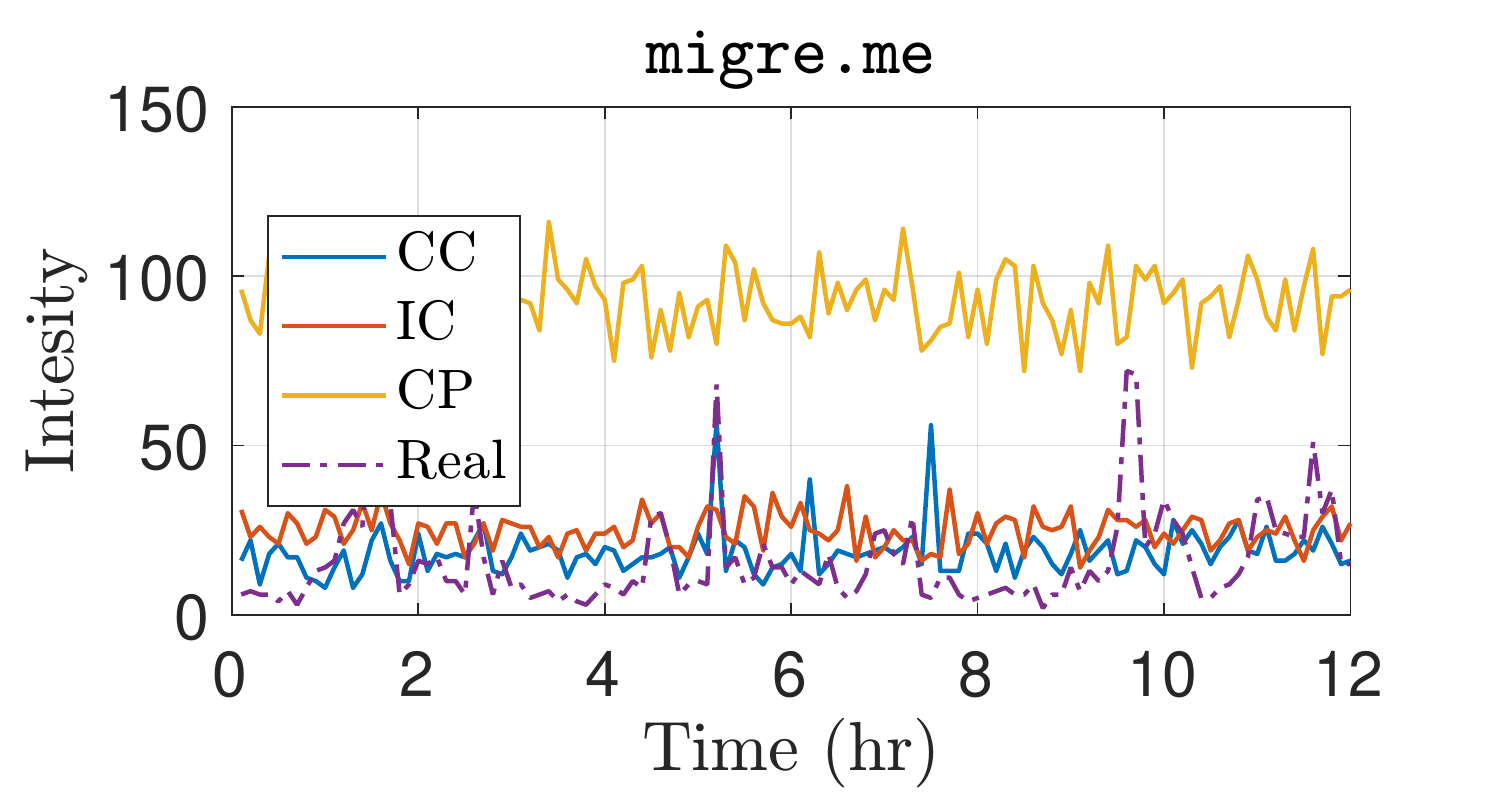} \hspace{-6mm}
\includegraphics[width=0.245\textwidth]{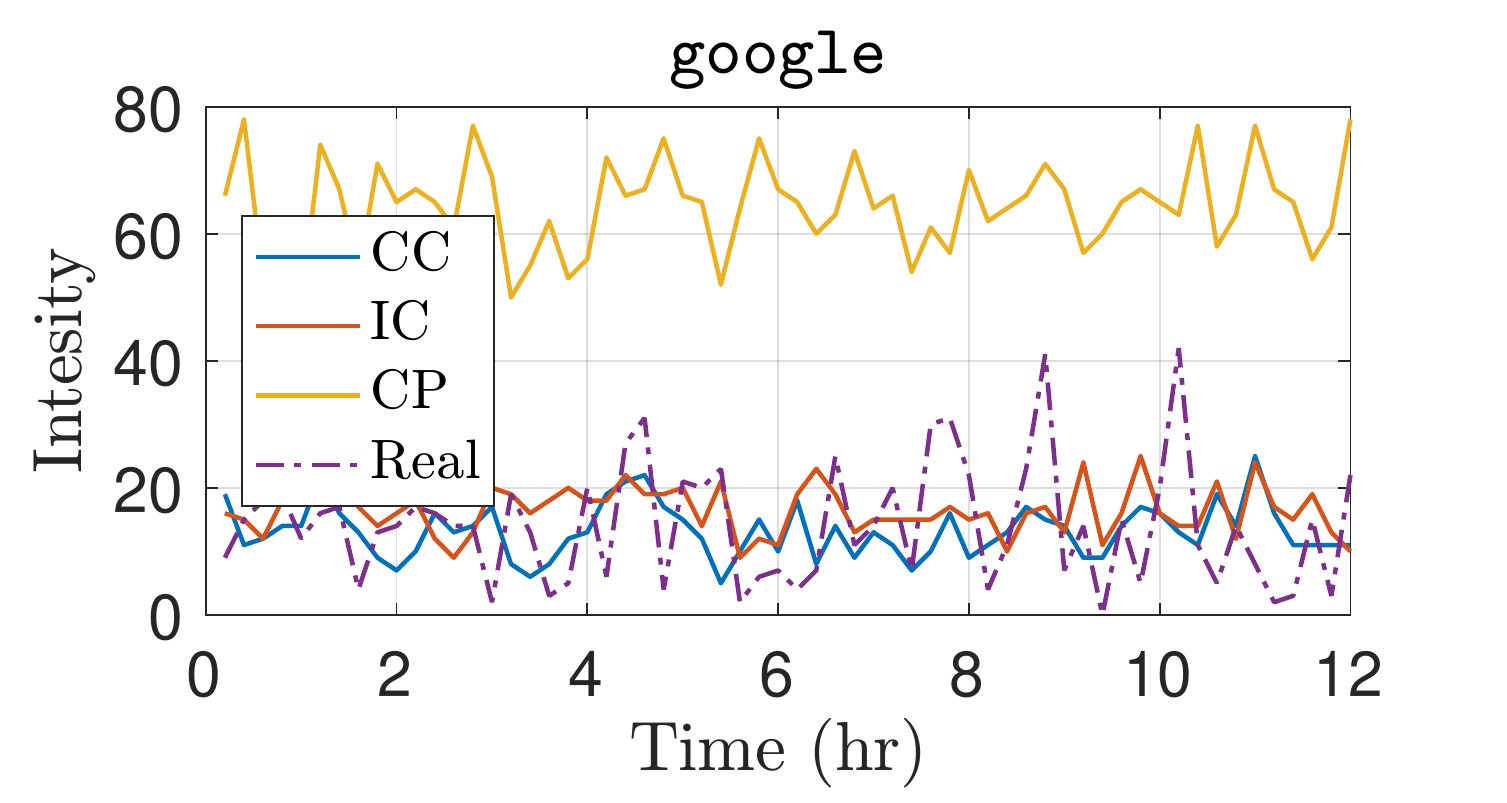} \vspace{-4mm}
\caption{Intensity of generated test events by three methods, compared with real test events in two exemplar products.}
\label{fig:testintensity}
\vspace{-3mm}
\end{figure}
We have also gathered our own dataset from Twitter. To select a set of active users, we query the Twitter search API, during one week in 2015, with some keywords about recent top music and singers. We select 30,000 users, that were actively tweeting about music and new albums.  Then all tweets of these users were crawled using Twitter API, during one month of 2015. 
To prune this dataset, the tweets containing the URLs of two popular media streaming services, Google Play Music and YouTube are retained. Then, we selected active users with more than 50 tweets. We refer to this dataset by ``Twitter music dataset".
The intensity (number of tweets per hour) of  URL and music datasets are plotted in Fig. \ref{fig:emp_music} in which  competition and cooperation between different cascades is apparent.

\begin{figure}[t]
\centering
\includegraphics[width=0.24\textwidth]{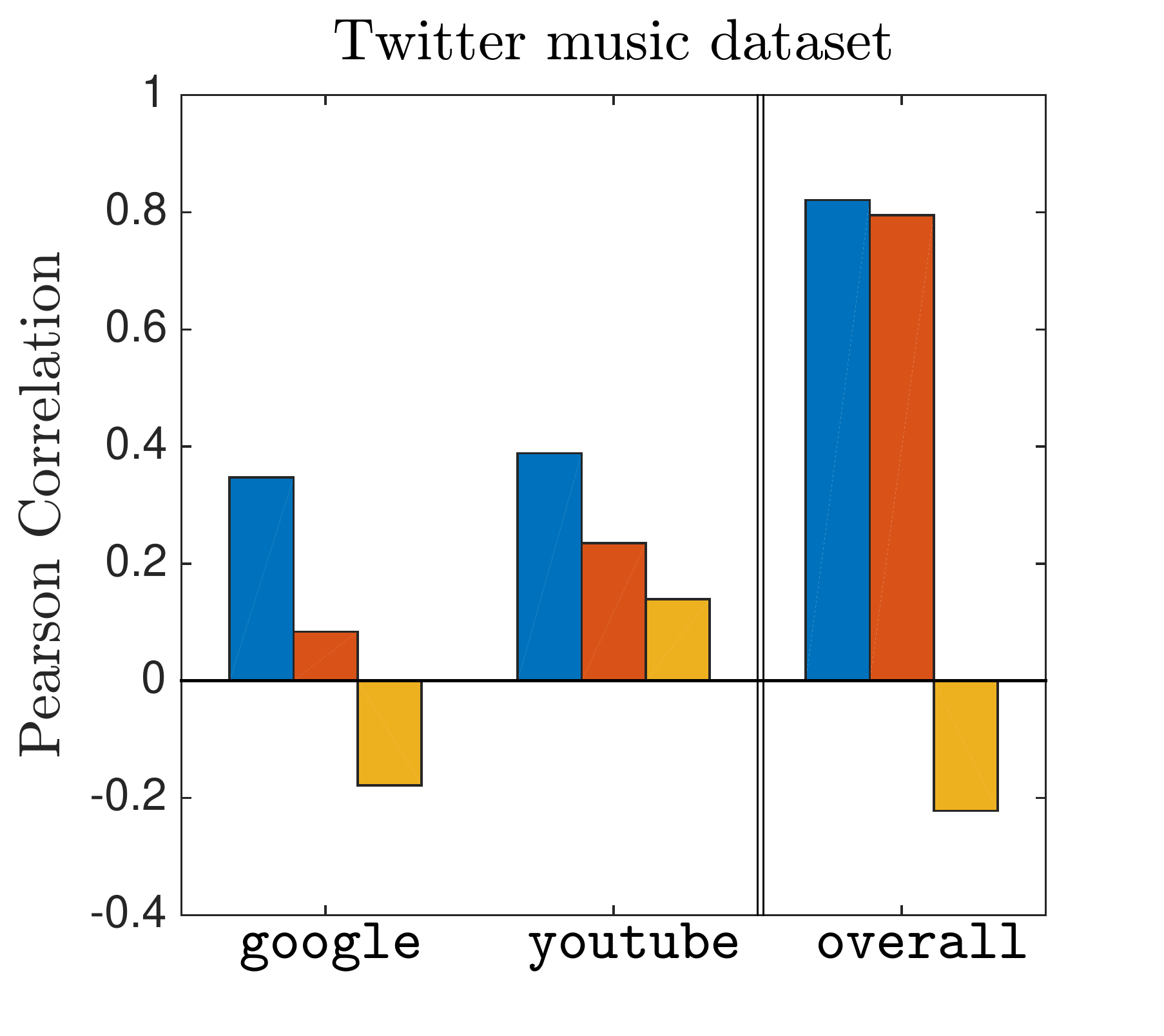} \hspace{-3mm}
\includegraphics[width=0.24\textwidth]{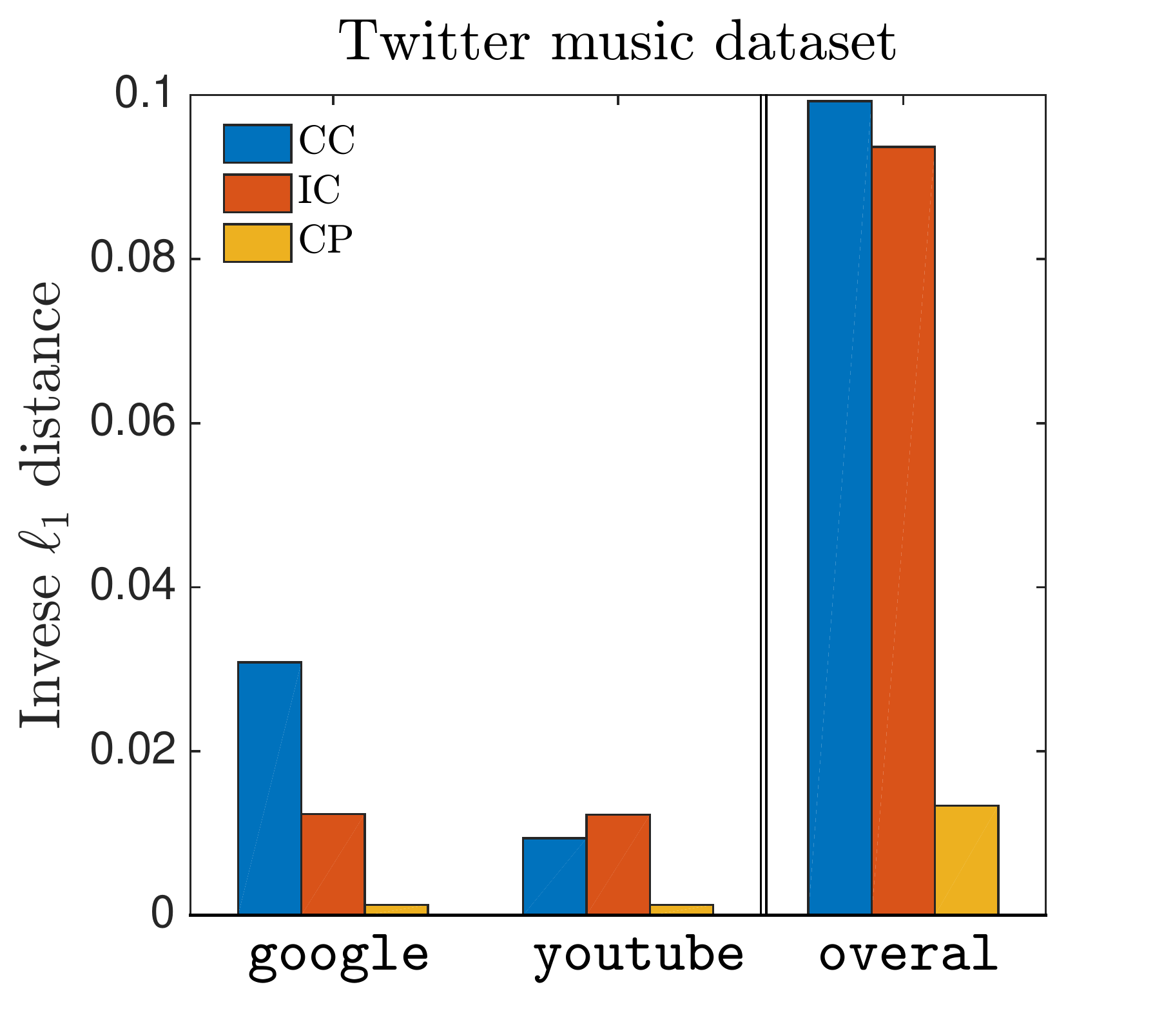}
\includegraphics[width=0.24\textwidth]{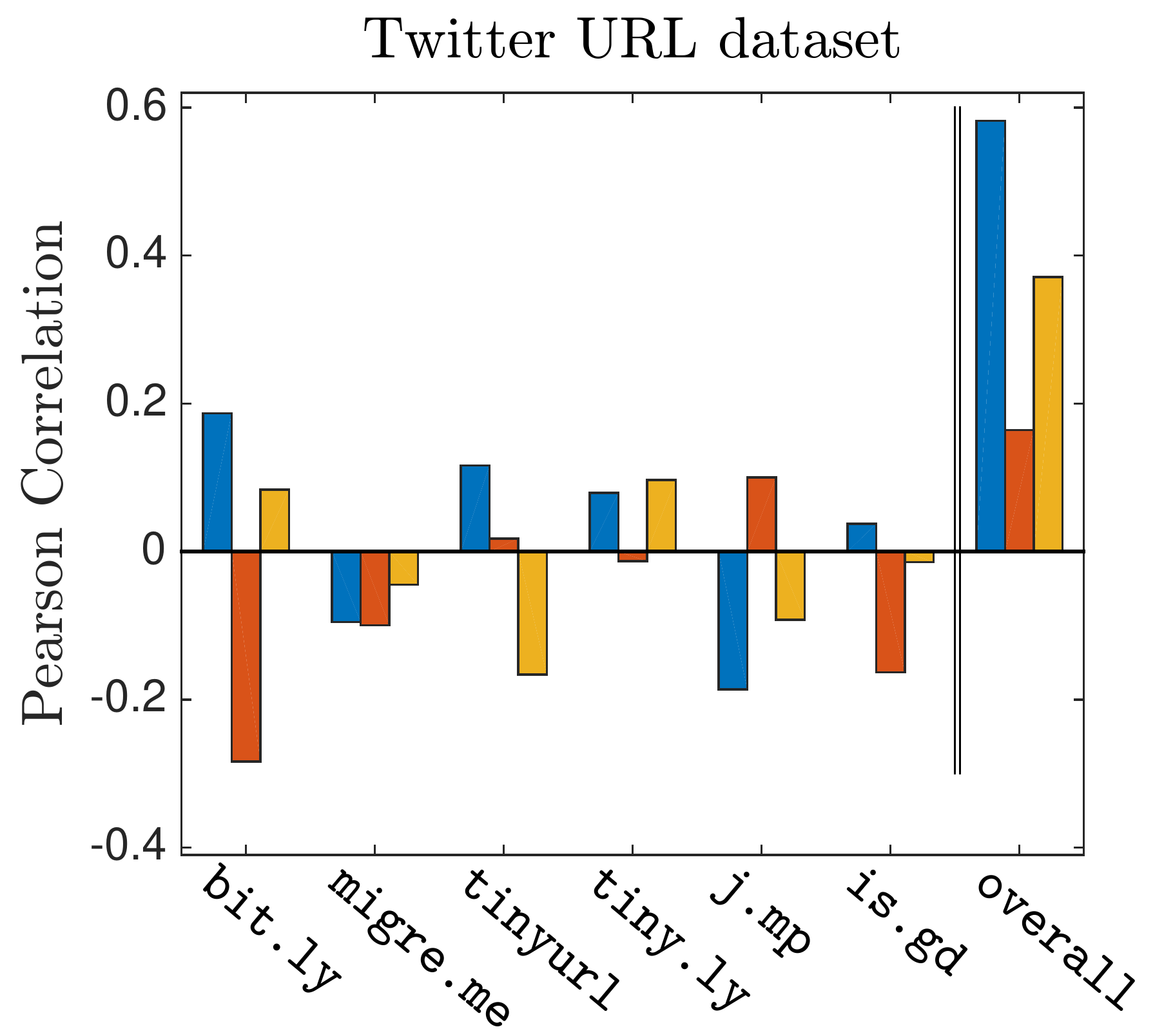}\hspace{-3mm}
\includegraphics[width=0.24\textwidth]{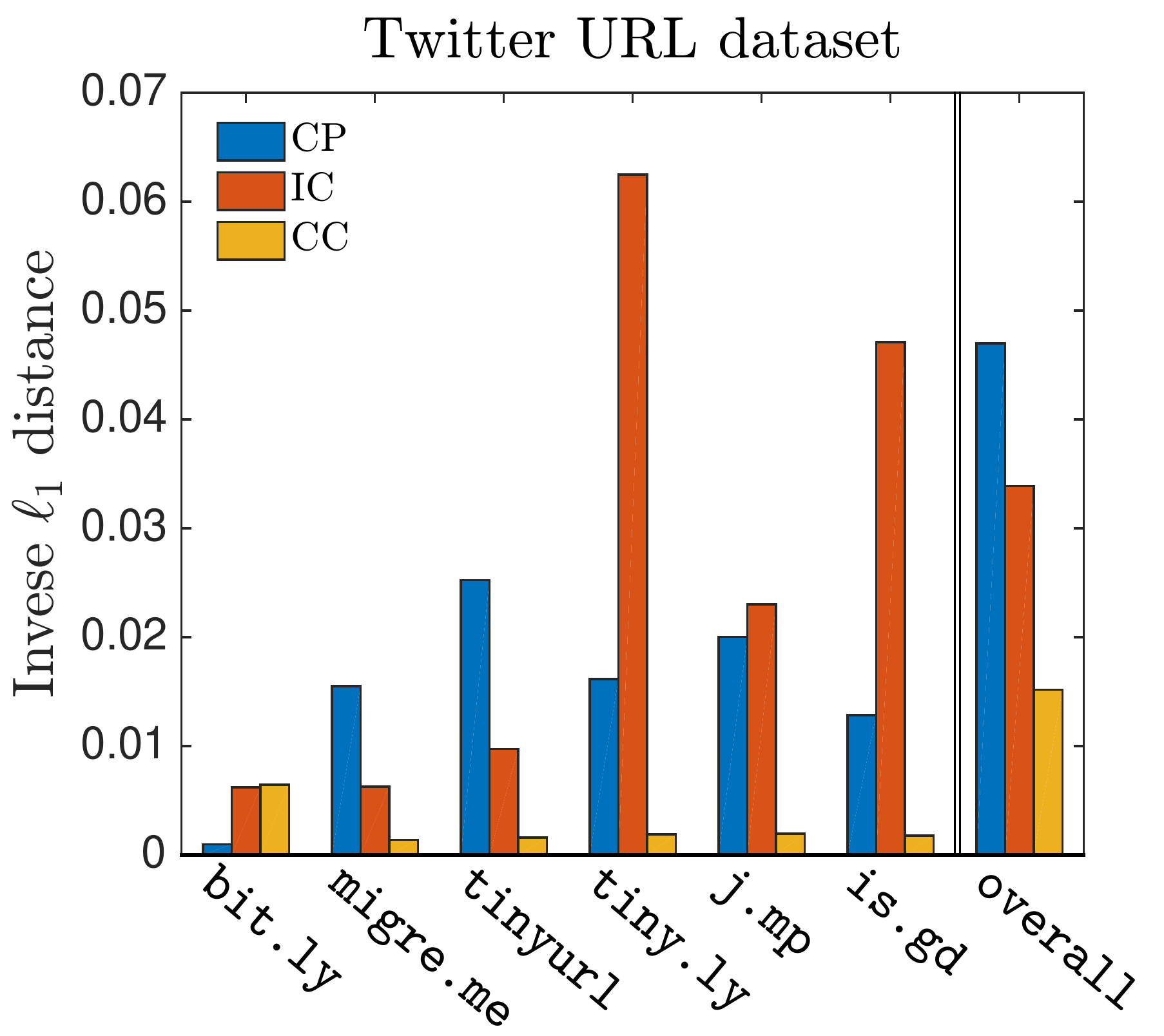} \vspace{-2mm}
\caption{Pearson correlation and inverse $\ell_1$ distance for the correlation of simulated test events.
\vspace{-3mm}
}
\label{fig:barplots}
\end{figure}

\textbf{Evaluation criteria}.
We evaluate our model in comparison with two other multiple cascade models. The names are abbreviated by \texttt{CC}, \texttt{IC}, and \texttt{CP}, respectively for Correlated Cascade,  Independent  Cascade and  Competing Products \cite{valera2015} models. In contrast to synthetic data, there is no ground-truth available for real datasets. Hence we use the \texttt{AvgPredLogLik}, Pearson correlation and $\ell_1$ distance which measure the prediction accuracy of the model.  

\textbf{Parameter Learning}.
We set aside the last $20\%$ of the data for the test set. The models are trained five times with $20\%$ to $100\%$ of the train data and $\beta$ found by cross-validation. The test likelihood for different models is plotted versus the training set size in Fig. \ref{fig:reallglks}. The proposed method has the highest likelihood in both datasets and is increasing with respect to the size of the training set. The weak performance of \texttt{CP} is  due to overfitting on the training data since the number of parameters in \texttt{CP} is proportional to the square of the number of products. Therefore, the overfit in music dataset (with 2 products) should be more severe than the URL dataset (with 6 products), which can be seen from Fig. \ref{fig:reallglks}. 
The slight decrease in the performance of the proposed method in music dataset is due to mix competitive-cooperative nature of this dataset. As illustrated in Fig. \ref{fig:emp_music}, there is a broad range of correlation between cascades. But even in this case, our model has better performance than \texttt{IC}.

\begin{figure}
\centering
\includegraphics[width=0.245\textwidth]{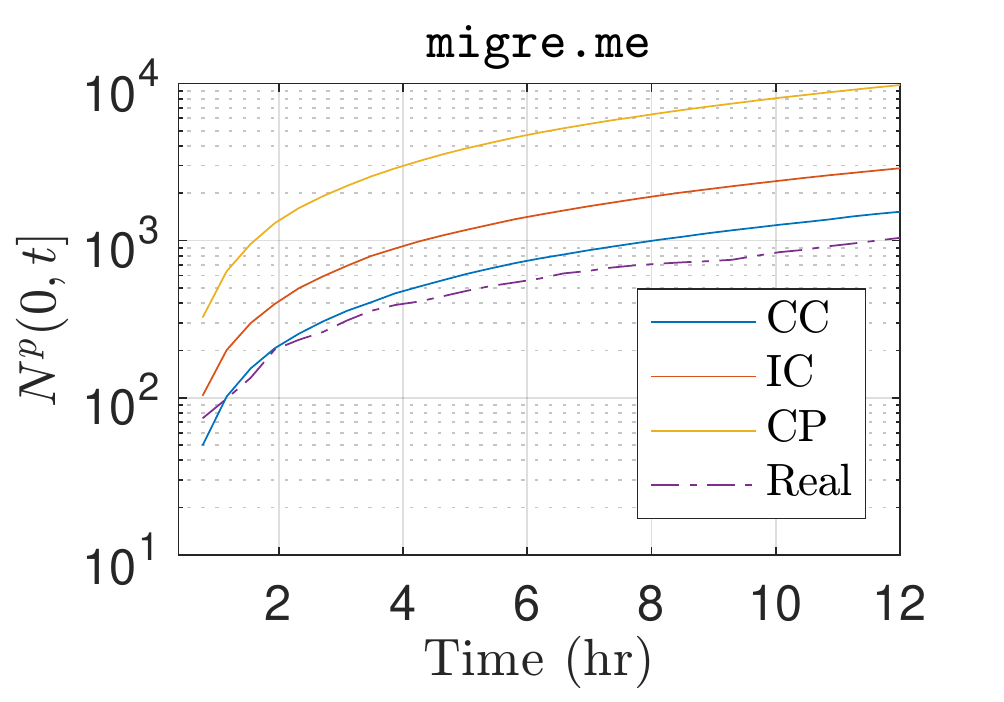}\hspace{-4mm}
\includegraphics[width=0.245\textwidth]{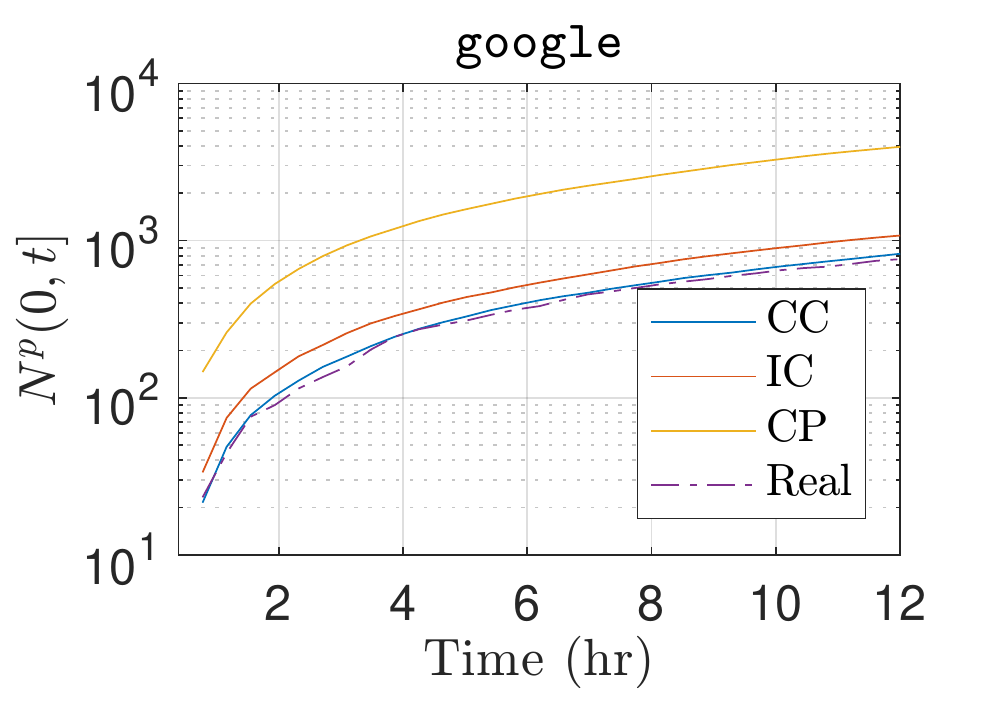} \vspace{-7mm}
\caption{Number of events for the two real exemplar products of Fig \ref{fig:testintensity} .}
\label{fig:testeventnum}
\vspace{-6mm}
\end{figure}

\textbf{Test Events Correlation}.
 To further investigate the proposed method, we also design some experiments on the simulated test events. Using the parameters of the learned model on the whole training data, we generate test events for each model. 
 The model that has higher correlation with real test event, is more successful to predict future events. We examined this feature, qualitatively and quantitatively in diagrams of Figs \ref{fig:testintensity} and \ref{fig:barplots}. We show only the intensity of one exemplar product for each dataset in Fig \ref{fig:testintensity}. 
Qualitatively it can be seen that the proposed method better followed the real test intensity curve, except in a few intervals like the times near 2, and 6 hours in the left of Fig \ref{fig:testintensity},  that real test intensity has large oscillations. Similar to the poor performance of \texttt{CP} in likelihood on test data, Fig \ref{fig:reallglks}, this model has generated more events, which results in its large distance with the true curve. 
Measuring the distance between two curves is a challenging problem in itself. We use two simple measures, the inverse of $\ell_1$ distance, and the well-know Pearson correlation. High inverse $\ell_1$ distance and Pearson correlation, indicated a high correlation between the two curves. In Fig \ref{fig:barplots} the performance of different methods on the two datasets is demonstrated. The result is plotted for two cases; separate products, and all products. In total, we have a higher correlation with the real test events. But like before, the performance of the proposed method in music dataset is slightly better than URL dataset, which is explained already. 
The number of test events for the mentioned exemplar product of Fig. \ref{fig:testintensity} is depicted in Fig. \ref{fig:testeventnum}. We use the semi-logarithmic scale in y-axis to better compare different methods. In both cases \texttt{CC} model results are the closest to the real test data.  

\vspace{-2mm}
\section{Conclusion}
\vspace{-2mm}
In this paper, we proposed a social \emph{behavior adoption} model in which multiple correlated cascades spread over the network. Multidimensional Hawkes process is utilized for the behavior or product adoption with its marks  capturing the decision making procedure of the users.
We have shown several properties of the proposed model on synthetic data. Furthermore, experiments on two real-world datasets establish the competitive-cooperative modeling capability and the superior performance of our model on predicting future events. Importantly, the parameter learning algorithm is shown to be quite efficient in both synthetic and real data.
 
For future work we would like to learn the hyperparameter and the decaying kernel.
Another interesting line of future work would be proposing a model to capture multiple cascades with mixed competing-cooperating behaviors.
\bibliographystyle{aaai}
\bibliography{ref}

\end{document}